\documentclass[a4paper,11pt]{article}
\usepackage[utf8]{inputenc}
\usepackage{geometry}
\usepackage{amssymb,amsmath,amsthm,xcolor}
\usepackage[
backend=biber,
style=numeric,
sorting=nyt,
sortcites,
giveninits=true,
url=false,
doi=false,
isbn=false,
date=year,
maxbibnames=4,
]{biblatex}
\renewbibmacro{in:}{}
\DeclareFieldFormat{pages}{#1}
\renewbibmacro*{volume+number+eid}{%
  \printfield{volume}%
  \setunit*{\addnbthinspace}
  \printfield{number}%
  \setunit{\addcomma\space}%
  \printfield{eid}}
\DeclareFieldFormat[article]{number}{\mkbibparens{#1}}
\usepackage[hidelinks]{hyperref}
\usepackage{authblk}
\usepackage{pgfplots}
\usepgfplotslibrary{fillbetween}
\usetikzlibrary{decorations.pathmorphing}
\usepackage{subcaption}
\usepackage{enumitem}

\title{Volume Singularities in General Relativity}
\author{Leonardo García-Heveling\thanks{\texttt{leonardo.garcia@uni-hamburg.de} \\ Radboud University Nijmegen, the Netherlands \\ Current address: University of Hamburg, Germany}}
\date{\today}

\newcommand{\vol}{\operatorname{vol}_g}
\newcommand{\area}{\operatorname{area}_g}
\newcommand{\volkb}{\operatorname{vol}_{\kappa,\beta}}
\newcommand{\areakb}{\operatorname{area}_{\kappa,\beta}}
\newcommand{\ccc}{\operatorname{CCC}(\kappa,\beta)}
\newcommand{\Ric}{\operatorname{Ric}}
\newcommand{\nat}{\mathbb{N}}
\newcommand{\real}{\mathbb{R}}

\newcommand{\scri}{\mathfrak{I}^+}

\newtheorem{thm}{Theorem}[section]
\newtheorem{lem}[thm]{Lemma}
\newtheorem{prop}[thm]{Proposition}
\newtheorem{conj}[thm]{Conjecture}
\newtheorem{cor}[thm]{Corollary}
\theoremstyle{definition}
\newtheorem{defn}[thm]{Definition}
\newtheorem{exam}[thm]{Example}

\begin{document}


\maketitle

\begin{abstract}
    We propose a new notion of singularity in General Relativity which complements the usual notions of geodesic incompleteness and curvature singularities. Concretely, we say that a spacetime has a volume singularity if there exist points whose future or past has arbitrarily small spacetime volume: in particular, smaller than a Planck volume. From a cosmological perspective, we show that the (geodesic) singularities predicted by Hawking's theorem are also volume singularities. In the black hole setting, we show that volume singularities are always hidden by an event horizon, prompting a discussion of Penrose's cosmic censorship conjecture.\\

\end{abstract}

\section{Introduction}

The concept of singularity is central in General Relativity. Already early on, it was discovered that certain solutions of the Einstein Equations feature unbounded components of the metric tensor. In some cases, such as the event horizon of Schwarzschild spacetime, this turned out to be due to a bad choice of coordinates. In other cases, such as the center of the Schwarzschild solution, a consensus developed that there is a ``real" singularity. Since then, an overwhelming amount of evidence for the existence of black holes, which are described by singular spacetimes such as Schwarzschild or Kerr, has been found. Furthermore, in cosmology there is a consensus that the whole Universe originates from an initial singularity, the Big Bang. Yet the nature and interpretation of singularities remains a debated and intricate topic, even the concept of singularity itself being rather loosely defined \cite{Ger}. In this paper, we offer a new point of view on the definition of singularities in General Relativity.

Our current understanding of singularities is largely based on the works of Penrose \cite{PenSing} (about black holes) and shortly after Hawking \cite{HawSing} (in the cosmological setting). The key idea is to consider a spacetime to be singular if it is incomplete, which could mean \cite[Chap.~6.2]{BEE}
\begin{itemize}
    \item Timelike geodesically incomplete,
    \item Null geodesically incomplete,
    \item Bounded acceleration incomplete.
\end{itemize}
Timelike geodesic incompleteness has the physical interpretation that there exists observers at rest who, after a finite proper time, reach the ``end" of spacetime. Similarly, bounded acceleration incompleteness has the same interpretation for observers that may be accelerated. Null geodesic incompleteness is more difficult to interpret, given that the affine parameter of a null geodesic does not correspond to an observable quantity.

The power of the concept of incompleteness lies in the singularity theorems of Hawking and Penrose, which show that generic spacetimes satisfying some physically reasonable conditions are geodesically incomplete (timelike in Hawking's theorem and null in Penrose's). The singularity theorems have since been improved in many ways: Their assumptions have been relaxed to account for more realistic physical scenarios  \cite{MinSing,GKOS,PaeSing}, and the differentiability of the spacetime metric has been lowered \cite{Graf16,GraSing,GGK}, showing that the predicted singularities are not just an artifact caused by assuming too much regularity. However, in the case of Penrose's theorem, so far no one has managed to strengthen the conclusion: while the theorem is commonly interpreted as to show the existence of a black hole, all it really predicts is the existence of an incomplete null geodesic, but it does not, for instance, predict the existence of a horizon.

In this paper, we propose a new type of incompleteness. Unlike the above ones, it is not based on the length of curves; instead, it makes use of the spacetime volume.
\begin{defn}
    A spacetime $(M,g)$ is \emph{future volume incomplete} if {for  every $\varepsilon >0$ there is a point $x \in M$ with $\vol(I^+(x)) < \varepsilon$.}
\end{defn}
Here $I^+(x)$ denotes the chronological future of the event $x$ \cite[p.~55]{BEE}. Similarly, we can define past volume incompleteness using the past $I^-(x)$. We also use the term \emph{volume singularity} in a more vague sense when talking about volume incomplete spacetimes.

{In a volume incomplete spacetime, there are thus events whose future has a volume smaller than a Planck (spacetime) volume.  Physically,} it is expected that on scales smaller than those of the Planck units, quantum effects become dominant. In particular, there are theories of quantum gravity which predict that spacetime is discrete at those scales \cite{BLMS}, or becomes effectively two-dimensional \cite{CarQG}. Therefore, an observer reaching a spacetime point $p$ with $\vol(I^+(p)) < V_P$ is no longer capable of making valid predictions using classical physics (here $V_P$ denotes the Planck spacetime volume). This signals a breakdown of General Relativity, and thus deserves being called a singularity.

{Theorem \ref{thm1} in Section \ref{sec:gen} tells us that \emph{a sufficient condition for a chronological spacetime to be volume incomplete is the existence of a single point $x$ with $\vol(I^+(x)) < \infty$}. This condition is more managable in practice, and allows us to easily show that} important geodesically singular spacetimes, such as Schwarzschild, are also volume singular. More generally, we prove that volume singularities are always hidden by an event horizon (except in the cosmological case where the whole spacetime originates from the singularity). This suggests that the notion of volume singularity accurately encodes our physical intuition about black holes, and leads us to reformulate Penrose's cosmic censorship conjectures into the statement that ``physically realistic singularities are volume singularities". We also give this a concrete meaning by conjecturing that under suitable assumptions, \emph{$\vol(I^+(S))<\infty$ for a trapped surface $S$}. The latter conjecture, if proven, would provide a ``volume" analogue of Penrose's singularity theorem, but with the much stronger conclusion that there is an event horizon, and hence really of a black hole.

We also apply the concept of volume singularity in the cosmological setting. There, we obtain a volume version of Hawking's singularity theorem, based on previous work by Treude and Grant \cite{TrGr13}. Moreover, we also construct a time function in the style of Andersson, Galloway, and Howard \cite{AGH98} for Big Bang spacetimes.

Finally, note that volume singularities, just as geodesic singularities, do not provide any information about the extendibility of the spacetime. In particular, one can construct artificial, unphysical singularities by removing a portion of a given spacetime. Hence one should only regard a spacetime to be truly singular if it has a singularity (in the volume or geodesic sense) \emph{and} it is inextendible. The most common criterion for inextendibility is the blow-up of some curvature invariant (such as the Kretschmann scalar), which moreover can be given a physical meaning in terms of tidal forces. Curvature blow-up, however, only ensures that the spacetime metric cannot be extended as a twice-continuously differentiable tensor (class $C^2$). In order to account also for weak solutions of the Einsten Equations, which can have regularity lower than $C^2$, other criteria have been developed \cite{Sbi,Sbi2,GLS,GrLi}, but the problem of determining (in)extendibility remains a hard one in general.
\medskip

\textbf{Outline.} The paper starts with a general discussion of volume incompleteness. We consider specific black hole spacetimes in Section \ref{sec:SchKerr}, followed by a more abstract treatment of ``volume" black holes and cosmic censorship in Section \ref{sec:CC}. In Sections \ref{sec:Haw} and \ref{sec:voltime} we turn to the cosmological setting, proving a Hawking style singularity theorem and constructing a cosmological volume function, respectively. We end the paper with a conclusions section summarizing the main results and open problems.

\section{General properties and examples of volume singularities} \label{sec:gen}

{We start by providing a more practical characterization of volume incompleteness for chronological spacetimes (that is, those without closed timelike curves).}

\begin{thm} \label{thm1}
A chronological spacetime $(M,g)$ is future volume incomplete if and only if it contains a point $x_0 \in  M$ such that $\vol(I^+(x_0)) < \infty$.
\end{thm}

{The proof makes use of the following fact, which will also be useful elsewhere in the paper, so we state it as a lemma.}

\begin{lem} \label{lem:subsetneq}
 {Let $x,y$ be two points such that $\vol\left(I^+(x)\right) < \infty$ and $I^+(y) \subsetneq I^+(x)$. Then} \[ \vol\left(I^+(y)\right) < \vol\left(I^+(x)\right). \]
\end{lem}

\begin{proof}
 {Choose a point $z \in I^+(x) \setminus I^+(y)$. Then $I^-(z) \cap I^+(y) = \emptyset$, since otherwise $z \in I^+(y)$ by transitivity. It follows that also $U := I^+(x) \cap I^-(z)$ is disjoint from $I^+(y)$. Since moreover $U$ is non-empty and open, it has non-zero volume. Therefore \[ \vol\left(I^+(y)\right) < \vol\left(I^+(y)\right) + \vol\left(U\right) \leq \vol\left(I^+(x)\right), \] completing the proof.}
\end{proof}

\begin{proof}[Proof of Theorem \ref{thm1}]
    Let
    \begin{align}
        v(x) &:= \vol\left(I^+(x)\right), \nonumber \\
        f(x) &:= \inf \big\{ v(z) \mid z \in I^+(x) \big\}, \label{eq:deff}
    \end{align}
    noting that for all $x \in M$ with $v(x) < \infty$ and $y \in I^+(x)$,
    \begin{align}
        v(y) &< v(x), \label{eq:vdec} \\
        f(y) &\geq f(x), \label{eq:fnondec} \\
        v(x) &> f(x). \label{eq:vgf}
    \end{align}
    Here \eqref{eq:vdec} holds {by Lemma \ref{lem:subsetneq}} since if $y \in I^+(x)$, then by chronology of $M$, $I^+(y) \subsetneq I^+(x)$ {(the inclusion follows by transitivity, and must be strict because $y \not\in I^+(y)$)}. But then the infimum in \eqref{eq:deff} for $y$ is over a smaller set than for $x$, proving also \eqref{eq:fnondec}. Finally, \eqref{eq:vgf} follows because one can always find a suitable $y \in I^+(x)$ with $v(x) > v(y) \geq f(x)$.

    Next, assume that there exists $x_0 \in M$ with $v(x_0) < \infty$. Then $f(x_0)<\infty$ by \eqref{eq:vgf}. If $f(x_0)=0$, we are done, so assume for the sake of contradiction that $f(x_0)>0$. Construct a sequence of points $(x_i)_{i \in \nat}$ such that $x_{i+1} \in I^+(x_i)$ and
    \begin{equation}
        v(x_{i+1}) \leq \frac{v(x_i) + f(x_i)}{2}. \label{eq:defseq}
    \end{equation}
    Such a sequence exists because, given $x_i$, \eqref{eq:deff}, \eqref{eq:vdec} and \eqref{eq:vgf} together allow one to find $x_{i+1}$. Since $f(x_i) \geq f(x_0) >0$, and by \eqref{eq:vdec} also $v(x_{i+1}) < v(x_i)$, we have that $v(x_i) \to v_\infty > 0$ as $i \to \infty$. Therefore
    \begin{equation*}
    \vol \left( \bigcap_{i \in \nat} I^+(x_i) \right) = \lim_{i \to \infty} v(x_i) > 0,
    \end{equation*}
    and hence
    \begin{equation*}
    P := \bigcap_{i \in \nat} I^+(x_i) \neq \emptyset.
    \end{equation*}
    Thus we can choose a point $x_\infty \in P$. Notice that by \eqref{eq:vdec}
    \begin{equation*}
        v(x_\infty) \ {\leq}\  \lim_{i \to \infty} v(x_i),
    \end{equation*}
    and by \eqref{eq:fnondec}
    \begin{equation*}
        f(x_\infty) \geq \lim_{i \to \infty} f(x_i),
    \end{equation*}
    where the limit exists because $f(x_i) < v(x_0) < \infty$ for all $i \in \nat$. On the one hand, {combining these two inequalities with \eqref{eq:vgf} shows that}
    \begin{equation*}
        \lim_{i \to \infty} v(x_i) \ { \geq v(x_\infty) > f(x_\infty) \geq }\ \lim_{i \to \infty} f(x_i).
    \end{equation*}
    On the other hand, taking the limit $i \to \infty$ in \eqref{eq:defseq} shows that
    \begin{equation*}
        \lim_{i \to \infty} v(x_i) \leq \lim_{i \to \infty} f(x_i).
    \end{equation*}
    A contradiction is reached, showing that in fact $f(x_0)$ can only be zero, and the theorem follows.
\end{proof}

We end this section with two examples to illustrate that, without any assumptions, the notions of geodesic and volume singularity are logically independent.

\begin{exam}[Geodesically complete but volume incomplete] \label{ex:volin}
Let $n \geq 3$, let $M = (1,\infty) \times S^n$ and let $h$ be the round metric on the $n$-sphere $S^n$. Equip $M$ with the Lorentzian metric
\begin{equation*}
    g = -dt^2 + t^{-1} h.
\end{equation*}
Then $(M,g)$ is future volume incomplete {by Theorem \ref{thm1}, since}
\begin{equation*}
    \vol \left( I^+(x)\right) < \vol\left(\{t> t_x\}\right) = \operatorname{vol}_h(S^n) \int_{t_x}^\infty t^{-\frac{n}{2}} dt < \infty
\end{equation*}
for all $x \in M$ (here $t_x$ denotes the value of the $t$-coordinate at $x$, and note the integral is finite since $n\ge 3$). Nonetheless, $(M,g)$ is future causally geodesically complete. While it is obvious that the $t$-lines are complete, it is less so that all causal geodesics are complete. Therefore, we proceed to compute the geodesics explicitly.

By symmetries of the sphere, we know that the projection onto $S^n$ of any geodesic $\gamma$ must remain restricted to the equator (when we choose the equator to be tangent to the projection of the initial tangent vector to $\gamma$). Moreover, the sphere admits a Killing vector field $V$ tangent to the equator, which we can choose to be normalized on the equator. The geodesics therefore admit a constant of motion
\begin{equation} \label{eq:ctmotion}
    L = g(\dot\gamma,V) = \frac{1}{t} \dot\theta,
\end{equation}
where $\theta$ denotes the natural angular coordinate along the equator, and the dot denotes differentiation with respect to the affine parameter of $\gamma$. In the cases of interest, $\gamma$ satisfies
\begin{equation} \label{eq:arclengthparam}
    -\delta = -\dot t^2 + \frac{1}{t} \dot\theta^2,
\end{equation}
for $\delta = 0,1$ ($\delta=0$ corresponds to a null geodesic, and $\delta=1$ to a timelike geodesic parametrized by arclength). Isolating $\dot\theta$ from \eqref{eq:ctmotion} and \eqref{eq:arclengthparam}, equating the ensuing expressions, and performing simple algebraic manipulations yields
\begin{equation*}
    \dot t = \sqrt{L^2 t + \delta}.
\end{equation*}
{Here we have chosen the square-root with a $+$ sign because we are only interested in future-directed geodesics.} We solve this ordinary differential equation by separation of variables, obtaining
\begin{equation} \label{eq:tofs}
    t(s) = \frac{L^2}{4} (s+C)^2 - \frac{\delta}{L^2}
\end{equation}
for some integration constant $C$. Substituting $t(s)$ into \eqref{eq:ctmotion} {and integrating}, we furthermore obtain
\begin{equation} \label{eq:thetaofs}
    \theta (s) = \frac{L^3}{12} \left( s+C \right)^3 - \frac{\delta}{L} s + D,
\end{equation}
{with another integration constant $D$ appearing.}

{We can establish the dependence of the constants $L,C,D$ on the starting point $(t_0,\theta_0)$ and the initial velocity $(\dot{t}_0,\dot{\theta}_0)$ by evaluating \eqref{eq:ctmotion}, \eqref{eq:tofs} and \eqref{eq:thetaofs} at $s=0$, obtaining}
\begin{equation*}
 L = \frac{\dot{\theta}_0}{t_0}, \qquad C = \frac{2}{L} \sqrt{t_0 + \frac{\delta}{L^2} }, \qquad D = \theta_0 - \frac{L^3 C^3}{12}. 
\end{equation*}
{Note that we had already fixed the norm of the tangent vector to the geodesic in \eqref{eq:arclengthparam}, which is why there are three and not four constants, and $\dot{t}_0$ does not appear explicitly in their expressions. However, \eqref{eq:tofs} only fixes $C^2$, and it is the fact that $\dot{t}_0 >0$ (as the geodesic is future directed) which forces us to choose $C>0$. Thanks to this, $t(s)$ is increasing and hence $t(s) \in (1,\infty)$ for all $s \in [0,\infty)$. Since neither \eqref{eq:tofs} nor \eqref{eq:thetaofs} blow up for finite affine parameter $s$, it follows that all causal geodesics are future-complete.}
\end{exam}

Conversely, there exist many geodesically incomplete spacetimes which are volume complete, such as maximally extended Kerr spacetime. However, the maximal extension of Kerr is not globally hyperbolic, whence we cooked up a new example to illustrate that global hyperbolicity is not an issue. Our example bears resemblance to the one in \cite[p.~531]{Ger}, but is both timelike and null incomplete.

\begin{figure}[h]
    \centering
    \begin{subfigure}{0.5\textwidth}
    \centering
    \begin{tikzpicture}
    \begin{axis}[xmin=-2, xmax=2, ymin=-2, ymax=2, hide axis,axis equal image,clip=false]
        \draw[fill=gray!20,gray!20] (axis cs: -2,-2) -- (axis cs: 2,2) -- (axis cs: 1,2) -- (axis cs: -2,-1); -- cycle;
        \draw (axis cs: -2,-2) -- (axis cs: 2,2) node [midway, above, sloped] {$u=0$};
        \draw (axis cs: -2,-1) -- (axis cs: 1,2) node [midway, above, sloped] {$u=1$};
        \node at (axis cs: -1.75,-0.25) {$A$};
        \node at (axis cs: -1.5,-1) {$B$};
        \node at (axis cs: -1.25,-1.75) {$C$};
        \addplot[domain=-0.75:5,samples=20,smooth] ({0.5*(x+1+exp(-x))},{0.5*(x-1-exp(-x))});
        \draw[-stealth] (axis cs:{0.5*(3.3+1+exp(-3.3))-0.005},{0.5*(3.3-1-exp(-3.3))}) -- (axis cs:{0.5*(3.4+1+exp(-3.4))-0.005},{0.5*(3.4-1-exp(-3.4))});
        \node at (axis cs:{0.5*(3.3+1+exp(-3.3))},{0.5*(3.3-1-exp(-3.3))-0.5}) {$\gamma$};
    \end{axis}
    \end{tikzpicture}
    \caption{The whole spacetime in $(u,v)$ coordinates.}
    \label{fig:geoina}
    \end{subfigure}%
    \begin{subfigure}{0.5\textwidth}
    \centering
    \begin{tikzpicture}
        \begin{axis}[xmin=-2, xmax=2, ymin=-2, ymax=2, hide axis,axis equal image,clip=false]
        \draw[dashed] (axis cs: -2,-2) -- (axis cs: 0,0)node [midway, above, sloped] {$u=0$};
        \draw[dashed] (axis cs: 0,0) -- (axis cs: 2,-2)node [midway, above, sloped] {$w=0$};
        \draw (axis cs:-0,-1) -- (axis cs:0.25,-0.25);
        \draw[-stealth] (axis cs:-0.33,-1.99) -- (axis cs:0.1,-0.7);
        \node at (axis cs:0.1,-1.2) {$\gamma$};
    \end{axis}
    \end{tikzpicture}
    \caption{The region $C$ in $(u,w)$ coordinates.}
    \label{fig:geoinb}
    \end{subfigure}
    \caption{The spacetime in Example \ref{ex:geoin}, with a future incomplete timelike geodesic $\gamma$.}
    \label{fig:geoin}
\end{figure}
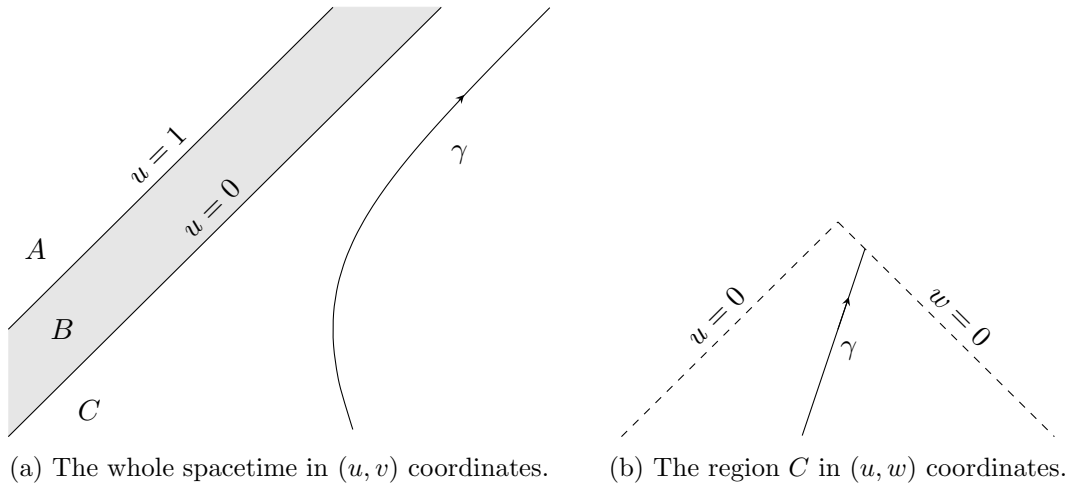

\begin{exam}[Geodesically incomplete but volume complete] \label{ex:geoin}
Consider $M = \real^2$ equipped with the Lorentzian metric
\begin{equation*}
 g = - \Omega(u,v) du dv,
\end{equation*}
which is conformal to the Minkowski metric written in null coordinates $u=t-x$, $v=t+x$. Hence $(M,g)$ is globally hyperbolic. We divide $M$ into three regions:
\begin{align*}
    A = \{(u,v) \mid 1<u\}, \quad
    B = \{(u,v) \mid 0 \leq u \leq 1\}, \quad
    C = \{(u,v) \mid u<0\},
\end{align*}
depicted in Figure \ref{fig:geoina}. We then set the conformal factor to
\begin{equation*}
    \Omega (u,v) = \begin{cases} 1 &\text{if } (u,v) \in A, \\ e^{-v} &\text{if } (u,v) \in C, \\ \text{smoothly extended} &\text{to } (u,v) \in B. \end{cases}
\end{equation*}
On the one hand, $(M,g)$ is future volume complete since for any $x \in M$, already $I^+(x) \cap A$ has infinite volume, so surely also $\vol(I^+(x)) = \infty$. On the other hand, region $C$ contains incomplete null and timelike geodesics. To see this, define a new coordinate $w=-e^{-v}$ on region $C$, with range $w \in (-\infty,0)$, so that the metric takes the form $g=-dudw$. Therefore region $C$ is isometric to a wedge of Minkowski spacetime, and all causal geodesics in $C$ are incomplete (see Figure \ref{fig:geoinb}). Those along which $u \to 0$ may enter region $B$ and are hence not necessarily incomplete in the whole spacetime $M$, but those along which $u \not\to 0$ are inextendible and hence incomplete also in $M$. See, for instance, the curve $\gamma$ in Figure \ref{fig:geoin}.
\end{exam}

\section{The singularities in Schwarzschild, Reissner--Nordstr\"om and Kerr spacetimes} \label{sec:SchKerr}

In this section, we discuss the most important black hole spacetimes in the context of volume singularities.

\begin{prop}
In Schwarzschild spacetime, points $p$ in the interior region (with $r_p < 2m$) have $\vol \left( I^+(p) \right) < \infty$, while points in the exterior region (with $r_p > 2m$) have $\vol \left( I^+(p) \right) = \infty$.
\end{prop}

\begin{proof}
We work in the usual coordinates where the expression of the Schwarzschild metric is
\begin{equation*}
    g = - \left( 1 - \frac{2m}{r} \right) dt^2 + \left( 1 - \frac{2m}{r} \right)^{-1} dr^2 + r^2 d \Omega.
\end{equation*}

\textbf{Interior.} A future-directed timelike curve $\gamma$ starting at $p_0 = (t_0,r_0,\theta_0,\phi_0)$ with $r_0 < 2m$ must satisfy
\begin{equation*}
    \left\vert 1 - \frac{2m}{r_0} \right\vert^2 \dot{t}^2 \leq \left\vert 1 - \frac{2m}{r} \right\vert^2 \dot{t}^2 < \dot{r}^2,
\end{equation*}
because $g(\dot\gamma,\dot\gamma) < 0$ and because $r$ is strictly decreasing along $\gamma$. Naming
\begin{equation*}
    C := \left\vert 1 - \frac{2m}{r_0} \right\vert^{-1},
\end{equation*}
we can write $|\dot{t}| < C|\dot{r}|$, which implies $t \in (t_0 - Cr_0, t_0 + Cr_0)$. We conclude that
\begin{align*}
    \vol\left( I^+(p_0) \right) &\leq \int_{t_0 - Cr_0}^{t_0 + Cr_0} \int^{r_0}_{0} 4 \pi r^2 dr dt = \frac{8}{3} \pi C r_0^4 < \infty.
\end{align*}

\textbf{Exterior.} Let $p_0 = (t_0,r_0,0,0)$ with $r_0 > 2m$. The curve
\begin{equation*}
    \gamma(s) = (t_0 + s,r_0,\left\vert 1 - \frac{2m}{r_0} \right\vert \frac{1}{r_0^2} s, 0)
\end{equation*}
is null, and for $s_1 := C r_0^2 \pi$ we have $\gamma(0) = p_0$ and $\gamma(s_1) = (t_0+s_1,r_0,\pi,0)$. By spherical symmetry, we conclude that every point of the form $(t_0+s_1,r_0,\theta,\phi)$ lies in $J^+(p_0)$, for all $\theta, \phi$. Furthermore, for $s_2 > s_1$, there exists $0 < \varepsilon < r_0 - 2m$ such that for all $r \in (r_0 - \varepsilon, r_0 + \varepsilon)$, $(t_0+s_2,r,\theta,\phi)$ is in the causal future of $(t_0+s_1,r_0,\theta,\phi)$. We conclude that the set
\begin{equation*}
    A := \left\{ (t,r,\theta,\phi) : t > t_0 + s_2, r \in (r_0 - \varepsilon, r_0 + \varepsilon) \right\} \subset J^+(p_0).
\end{equation*}
But
\begin{equation*}
    \vol(A) = \int_{t_0+s_2}^\infty \int_{r_0 - \varepsilon}^{r_0+\varepsilon} 4 \pi r^2 dr dt = \infty,
\end{equation*}
and therefore also $\vol\left( I^+(p) \right) = \vol\left( J^+(p) \right) = \infty$.
\end{proof}

{Recall now the Reissner--Nordstr\"{o}m metric
\begin{equation*}
    g = - \left( 1 - \frac{2m}{r} + \frac{e^2}{r^2} \right) dt^2 + \left( 1 - \frac{2m}{r} + \frac{e^2}{r^2} \right)^{-1} dr^2 + r^2 d \Omega.
\end{equation*}
The key difference compared to Schwarzschild is that here the factor in front of $dt^2$ has two zeros $r_{\pm}$. In the sub-extremal case (meaning $\vert e \vert < m$), $0 < r_- < r_+$ are distinct and both lead to coordinate singularities. The hypersurface $\{ r = r_+ \}$ is the event horizon, while $\{ r = r_- \}$ is known as the Cauchy horizon. The portion $\real \times (r_-,\infty) \times S^2$ of the spacetime is globally hyperbolic and hence uniquely determined by the Einstein Equations given an initial data hypersurface. Beyond $r_-$, the coordinate expression of the metric still makes sense (until it reaches a curvature singularity at $r=0$). This extension, however, is rather arbitrary since it is no longer globally hyperbolic and hence not uniquely determined. Moreover, it contains another Cauchy horizon beyond which one can extend further, for instance by periodically ``stacking'' copies of the spacetime (this gives the maximal analytic extension, see \cite[Chap.\ 9]{Klaas}).}

\begin{prop} \label{prop:Kerr}
 {The globally hyperbolic region $\{ r > r_- \}$ of Reissner--Nordstr\"{o}m spacetime is volume incomplete, with $\vol \left( I^+(p) \right) < \infty$ for all $p$ with $r > r_+$ (i.e.\ all $p$ lying beyond the event horizon). On the other hand, the maximal analytic extension of Reissner--Nordstr\"{o}m spacetime is volume complete.}
\end{prop}

\begin{proof}
 {Let $p \in \{r_- < r < r_+ \}$. Consider the chronological future $I_{r>r_-}^+(p)$ in the region $\{ r > r_- \}$, and take its closure in the full (extended) spacetime, as depicted in Figure \ref{fig:Pendiag}. The closure is compact, and therefore it has finite volume (w.r.t.\ the volume measure on the extended spacetime). But then the set $I_{r>r_-}^+(p)$ in the un-extended spacetime must also have finite volume. Notice here that Penrose diagrams are constructed using conformal diffeomorphisms (which in particular are homeomorphisms and hence preserve compactness), and that in the picture we supress the radial coordinates (but of course $S^2$ is compact). In the maximal analytic extension, on the other hand, since from any point it is possible to avoid the curvature singularity and enter the next periodic region, clearly the volume of the future is infinite for any point.}
\end{proof}

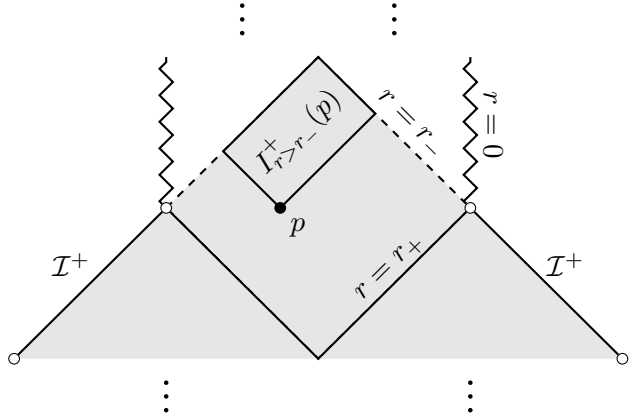
\begin{figure}
    \centering
    \begin{tikzpicture}
    \node[rotate=90] at (-1,4.5) {\huge{...}};
    \node[rotate=90] at (1,4.5) {\huge{...}};
    \node[rotate=90] at (-2,-0.5) {\huge{...}};
    \node[rotate=90] at (2,-0.5) {\huge{...}};
    \draw[decorate,decoration=zigzag,thick] (-2,2) -- (-2,4);
    \draw[decorate,decoration=zigzag,thick] (2,2) -- (2,4);
    \draw[dashed,fill=gray!20,thick] (-2,2) -- (0,4) -- (2,2) -- (0,0) -- cycle;
    \path[fill=gray!20] (-4,0) -- (-2,2) -- (0,0) -- (2,2) -- (4,0) -- cycle;
    \draw[thick] (-4,0) -- (-2,2) -- (0,0) -- (2,2) -- (4,0);
    \node at (3.25,1.25) {$\scri$};
    \node at (-3.25,1.25) {$\scri$};
    \node[rotate=315] at (1.2,3.1) {$r = r_-$};
    \node[rotate=45] at (0.95,1.2) {$r = r_+$};
    \node[rotate=270] at (2.3,3.1) {$r = 0$};
    \filldraw[white] (-4,0) circle (2pt);
    \draw (-4,0) circle (2pt);
    \filldraw[white] (4,0) circle (2pt);
    \draw (4,0) circle (2pt);
    \filldraw[white] (-2,2) circle (2pt);
    \draw (-2,2) circle (2pt);
    \filldraw[white] (2,2) circle (2pt);
    \draw (2,2) circle (2pt);
    \draw[fill=gray!20,thick] (-0.5,2) -- (-1.25,2.75) -- (0,4) -- (0.75,3.25) -- cycle;
    \node[rotate=45] at (-0.25,3) {$I_{r > r_{-}}^+(p)$};
    \filldraw[black] (-0.5,2) circle (2pt) node[anchor=north west] {$p$};
    \end{tikzpicture}
    \caption{A part of the Penrose diagram of Reisner--Nordstr\"om spacetime. The point $p$ and its future $I_{r > r_-}^+(p)$ within the globally hyperbolic region (shaded) are chosen as in the proof of Proposition \ref{prop:Kerr}.}
    \label{fig:Pendiag}
\end{figure}

{Notice that the same conclusion holds for sub-extremal Kerr spacetime, which has the same structure of an event horizon, beyond it a Cauchy horizon, and beyond that a curvature singularity (the so-called ring singularity). Beware that in the case of Kerr, there are closed timelike curves beyond the Cauchy horizon, making Theorem \ref{thm1} inapplicable. In any case, we conclude that the concept of volume incompleteness is able to capture the incompleteness of the maximal Cauchy developments, but it is unable to detect the curvature singularities in the extensions. Whether this is a drawback or not is debatable, since according to the strong cosmic censorship conjecture, anything lying beyond the Cauchy horizon is unphysical. In particular, the Cauchy horizon is believed to be unstable, as illustrated by Penrose's blue-shift argument (see \cite[Chap.\ 10]{Klaas}). In the next section, we discuss in a more abstract context how the notion of volume incompleteness can be related to cosmic censorship.}

\section{General black hole spacetimes} \label{sec:CC}

{In order to study black holes in an abstract way, we first need to give a definition of black hole that appropriately encodes the intuitive idea of a spacetime region that cannot be escaped. The usual approach is to define what it means to ``escape'' by defining future null infinity $\scri$, so that its chronological past $I^-(\scri)$ corresponds to the possible starting points of observers escaping to infinity. Hence its complement $M \setminus I^-(\scri)$ must be the black hole, and the boundary between the two, the event horizon. Traditionally, $\scri$ is defined as a subset of the conformal boundary of the spacetime, i.e.\ one conformally embeds the spacetime into a larger manifold. This raises the question of existence and uniqueness of such embeddings, which has been addressed by Chru\'sciel \cite{ChrExt}. One can also avoid this question altogether by defining $\scri$ via the causal boundary (instead of the conformal one), as has been done by Costa e Silva, Flores, and Herrera \cite{CFH1,CFH2}. A completely different way to define black hole is to instead formalize what it means to ``fall into the singularity'', which is the philosophy behind the works of M\"uller \cite{MueHor} (using geodesic incompleteness) and Wheeler \cite{Whe22} (via Scott and Szekeres' abstract boundary). We develop our own ``volumetric'' approach in this section.}

{Consider the subset of spacetime given by
\begin{equation*}
 \mathcal{B}  := \{x \in M \mid \vol \left( I^+(x) \right) < \infty \}.
\end{equation*}
If $(M,g)$ is chronological and $\mathcal{B} \neq \emptyset$, then by Theorem \ref{thm1}, $(M,g)$ is volume incomplete. We interpret $\mathcal{B}$ as the union of the interiors of all ``volume" black holes in our spacetime. For ease of exposition, we will assume that there is only one such black hole (i.e.\ that $\mathcal{B}$ is connected). By transitivity of the chronological relation, we immediately obtain the following result.}

\begin{prop} \label{prop:horizon}
The set $\mathcal{B}$ is a future set, {i.e.\ $I^+(\mathcal{B}) \subseteq \mathcal{B}$.}
\end{prop}

{Proposition \ref{prop:horizon} tells us that observers inside of the region $\mathcal{B}$ cannot leave it. Thus the boundary $\partial \mathcal{B}$ corresponds to the event horizon (in the sense of a one-way membrane bounding the black hole interior). For connected spacetimes $M$, we have $\partial B \neq \emptyset$ unless $\mathcal{B}=M$. The latter case} would correspond to a big crunch (or an ultra massive spacetime \cite{SenUM}) rather than a black hole, since it would mean that the entire Universe is swallowed by the singularity. We thus concentrate on the case $\mathcal{B} \neq M$ in this section.

{Recall that Penrose's Weak Cosmic Censorship (WCC) conjecture informally states that singularities in physically realistic spacetimes cannot be globally naked, meaning  that they must be hidden by an event horizon. This is thus satisfied for volume singularities, in the above sense. Penrose's Strong Cosmic Censorship (SCC) conjecture further states that singularities in physically realistic spacetimes cannot be locally naked, meaning that the spacetime must be globally hyperbolic (see \cite[Chap.\ 10]{Klaas} for a detailed discussion of cosmic censorship). The following proposition naturally characterises global hyperbolicity of the region $\mathcal{B}$.}

\begin{prop} \label{Prop:Bisgh}
    {The following are equivalent:
    \begin{enumerate}
     \item For every future-inextendible future-directed causal curve $\gamma \colon [0,b) \to M$ whose starting point $x = \gamma(0)$ satisfies $\vol(I^+(x)) < \infty$, we have \[\lim_{s \to b} \vol\Big(I^+\big(\gamma(s)\big)\Big) = 0.\]
     \item The region $\mathcal{B} := \{x \in M \mid \vol \left( I^+(x) \right) < \infty \}$ is globally hyperbolic.
    \end{enumerate}}
\end{prop}

\begin{proof}
    $1 \implies 2$. The region $\mathcal{B}$ can be equipped with a time function $\tau(x) := \vol(I^+(x))$, which is just a time-reversed version of the regular cosmological volume function considered in Section \ref{sec:voltime}. It then follows from Theorem \ref{thm:voltimegh} that $\mathcal{B}$ is globally hyperbolic.

    {$2 \implies 1$. Assume $1$ is not true, meaning that there is some inextendible causal curve $\gamma\colon [0,b) \to M$ with $\vol(I^+(\gamma(0))) < \infty$ but  $\vol(I^+(\gamma(s))) \not\to 0$. By transitivity of the timelike relation, we have
    \begin{equation*}
    \vol \left( \bigcap_{s \in [0,b)} I^+(\gamma(s)) \right) = \lim_{s \to b} \vol\Big(I^+(\gamma(s))\Big) > 0,
    \end{equation*}
    where in particular the limit exists (but is non-zero by assumption). Thus we can choose a point $y \in \bigcap_{s \in [0,b)} I^+(\gamma(s)) \neq \emptyset$. But then the causal diamond $J(\gamma(0), y)$ contains the inextendible curve $\gamma \colon [0,b) \to M$, hence $J(\gamma(0), y)$ is non-compact, in contradiction to global hyperbolicity.}
\end{proof}

{We are left to wonder whether the singularities appearing in physically realistic spacetimes are volume singularities, and how to even make such a statement precise. Indeed, turning the cosmic censorship conjectures of Penrose into precise mathematical conjectures is not so straightforward. The most popluar approach when it comes to WCC, due to Christodoulou \cite{ChrSCC}, is to define ``physically realistic spacetime'' as a solution of the Einstein equations for generic asymptotically flat initial data, and ``possesing an event horizon'' as the Maximal Globally Hyperbolic Development (MGHD) of said initial data possesing a complete null infinity $\scri$ (cf. beginning of this section). Similarly, SCC from this point of view requests the MGHD to be inextendible (since it is globally hyperbolic by definition, so the question is rather if it ``is everything there is'').}

{In an attempt to use volume incompleteness as a tool to formalize the cosmic censorship conjectures (\`a la Penrose), a logical first step is to try to establish a ``volume version'' of Penrose's singularity theorem, which could read as follows.}

\begin{conj} \label{conj:Penvolsing}
    Let $(M,g)$ be {a globally hyperbolic spacetime satisfying the strong energy condition}, and $S \subset M$ a future trapped surface. Then $\vol\left(I^+(S)\right) < \infty$.
\end{conj}

{Penrose's theorem shows that singularities (in the geodesic sense) form under generic conditions, and Conjecture \ref{conj:Penvolsing} would further show that an event horizon is formed. While the assumptions are tentative, we have strengthened the null energy condition from Penrose's theorem to the strong energy condition. The latter gives a much stronger control over the volumes, in particular allowing for volume comparison results. We use this in Section \ref{sec:Haw} to obtain a volume analogue of the Hawking singularity theorem. Penrose's theorem, however, is more difficult to adapt given the fact that a trapped surface has codimension $2$ (while in Hawking's theorem we have a codimension $1$ submanifold).}

{A more ambitious programme is to use one of the known characterizations of singularities (e.g.\ causal geodesic incompleteness), and attempt to prove that under physically realistic assumptions, these are also volume singularities.}

\begin{conj}[Volumetric WCC] \label{conj:volWCC}
    {In a generic physically realistic spacetime, every future incomplete causal geodesic enters the set $\mathcal{B}$.}
\end{conj}

\begin{conj}[Volumetric SCC] \label{conj:volSCC}
    In a generic, physically realistic, and inextendible spacetime,
    \begin{equation*}
     \vol\Big(I^+(\gamma(s))\Big) \to 0
    \end{equation*}
    along every future incomplete causal geodesic.
\end{conj}

{The names are chosen because Conjecture \ref{conj:volWCC} predicts the existence of an event horizon in the sense on Proposition \ref{prop:horizon}, while Conjecture \ref{conj:volSCC} is stronger and implies global hyperbolicity of the black hole interior by Proposition \ref{Prop:Bisgh}. We have formulated these conjectures in a bold way, but one can also think of them as criteria for the existence of an event horizon and for global hyperbolicity, if instead of considering all ``physically realistic'' spacetimes, one restricts to more specific situations. In any case, it of course remains to specify what ``physically realistic" means. In order to rule out trivial counterexamples (e.g.\ Minkowski with some region removed), one should assume either inextendibility of the spacetime, or that it is a MGHD of some initial data. Since Conjecture \ref{conj:volSCC} aims to predict global hyperbolicity, the second option would make it trivial, so we have formulated it required inextendibility. A weaker causality condition, such as causal simplicity, is still appropriate. One should also assume that the Einstein Equations hold for some reasonable matter model, either in a strict sense, or through energy conditions. Finally, the genericity assumption is needed (as it is in the IVP formulation), in order to rule out Kerr or Reissner--Nordstr\"{o}m as counterexamples. Again, it remains to specify exactly what ``generic'' means.}

\section{The cosmological volume singularity theorem} \label{sec:Haw}

In this section, we prove that the geodesic singularities predicted by Hawking's theorem are also volume singularities. Our proof is based on the work of Treude and Grant \cite{TrGr13}, where an alternative proof for Hawking's theorem is given, using volume comparison techniques. See also Graf \cite{Graf16} for a review of volume comparison results in the same spirit (assuming only $g \in C^{1,1}(M)$), and for the version of Hawking's theorem which we generalize \cite[Thm.~4.2]{Graf16}. This is an extended version thereof, in the sense that it allows for the Ricci tensor to be bounded from below by any constant (instead of zero, as in the original formulation by Hawking), as long as one adapts the mean curvature bound on the initial Cauchy surface accordingly (see also \cite[Thm.~1.2]{GaWo}).

We assume all Cauchy surfaces to be smooth and spacelike. Otherwise we follow the conventions of \cite{TrGr13}.

\begin{defn}[{\cite[Def.~5]{TrGr13}}] \label{def:CCC}
Let $(M,g)$ be a $(n+1)$-dimensional spactime, $\Sigma \subset M$ a Cauchy surface and $\kappa,\beta \in \real$ be constants. We say that $(M,g,\Sigma)$ satisfies the \emph{Cosmological Comparison Condition} $\ccc$ if
\begin{enumerate}
    \item $\Ric(v,v) \geq n \kappa$ for all $v \in TM$ with $g(v,v) = -1$,
    \item The mean curvature $H$ of $\Sigma$ satisfies $H \leq \beta$.
\end{enumerate}
\end{defn}

\begin{thm} \label{thm:sing}
Let $(M,g)$ be a globally hyperbolic $(n+1)$-dimensional spacetime and $\Sigma \subset M$ a Cauchy hypersurface with $\area(\Sigma) < \infty$. Assume that $M$ and $\Sigma$ satisfy the $\ccc$ with $\beta < 0$ and $\kappa \geq -\left(\beta/n\right)^2$, or with $\kappa > 0$ and any $\beta \in \real$. Then $\vol(I^+(\Sigma)) < \infty$ {and $(M,g)$ is volume incomplete}.
\end{thm}

{A special case of $\Sigma$ having finite area is when $\Sigma$ is compact. Note also that if we only assume that the projection onto $\Sigma$ (following the flow of the normal vector field) of $I^+(x)$ has finite area for some $x \in I^+(\Sigma)$, we can still conclude that $\vol(I^+(x)) < \infty$.}

The proof proceeds by comparing the volumes in our spacetime $(M,g)$ to volumes in the model space $M_{\kappa,\beta} = (a_{\kappa,\beta},b_{\kappa,\beta}) \times N_{\kappa,\beta}$. Here $N_{\kappa,\beta}$ is one of the three Riemannian $n$-dimensional simply connected spaces of constant curvature ($\mathbb{S}^n,\real^n$, or $\mathbb{H}^n$). We equip $M_{\kappa,\beta}$ with the Lorentzian metric
 \begin{equation*}
     g_{\kappa,\beta} := -dt^2 + f_{\kappa,\beta}^2 h_{\kappa,\beta}
 \end{equation*}
 where $h_{\kappa,\beta}$ is the Riemannian metric on $N_{\kappa,\beta}$ and
 \begin{equation*}
     f_{\kappa,\beta} \colon (a_{\kappa,\beta},b_{\kappa,\beta}) \to \real.
 \end{equation*}
 The manifold $N_{\kappa,\beta}$, the (possibly infinite) interval $(a_{\kappa,\beta},b_{\kappa,\beta})$ and the function $f_{\kappa,\beta}$ are determined by the values of $\kappa,\beta$ as in \cite[Table 1]{Graf16}.
 
\begin{proof}[Proof of Theorem \ref{thm:sing}]
\textbf{Case $\beta < 0,$ $\kappa = -(\beta/n)^2$.} In this case
\begin{equation*}
     f_{\kappa,\beta} = \exp\left(-\sqrt{|\kappa|} t\right),
 \end{equation*}
 and in particular $f_{\kappa,\beta}$ is defined for all $t \in \real$. For any $B \subset \{f_{\kappa,\beta} =0\}$, \cite[Eqns.\ (6) \& (15)]{TrGr13} imply that
 \begin{align*}
     \volkb B^+_B (t) &= \frac{\areakb B}{f_{\kappa,\beta} (0)^n} \int_0^t f_{\kappa,\beta} (s)^n ds = \frac{\areakb B}{2n \sqrt{|\kappa|}} \left(1 - \exp\left(-\sqrt{|\kappa|} t\right) \right),
 \end{align*}
 where $B^+_B(t) = (0,t) \times B$. Then by \cite[Thm.\ 9]{TrGr13}
 \begin{equation*}
     \vol I^+(\Sigma) = \lim_{t \to \infty} \vol B^+_\Sigma (t) \leq \frac{\area \Sigma}{\areakb B} \lim_{t \to \infty} \volkb B^+_B (t) = \frac{\area \Sigma}{2n \sqrt{|\kappa|}} < \infty,
 \end{equation*}
 where $B^+_\Sigma (t)$ is the image of $(0,t)$ under the flow of the geodesics normal to $\Sigma$.
 
 \textbf{All other cases.} By \cite[Thm.\ 4.2]{Graf16}, $(M,g)$ every future-directed timelike geodesic starting from $\Sigma$ is incomplete, with uniform bound on the lengths equal to $b_{\kappa,\beta} < \infty$. Therefore
 \begin{equation*}
     I^+(\Sigma) = B^+_\Sigma(b_{\kappa,\beta}),
 \end{equation*}
 and by \cite[Thm.\ 9]{TrGr13}
 \begin{equation*}
     \vol B^+_\Sigma(b_{\kappa,\beta}) \leq \frac{\area \Sigma}{\areakb B} \lim_{t \to b_{\kappa,\beta}} \volkb B^+_B (t) < \infty,
 \end{equation*}
 where finiteness follows from a straightforward computation for each possible $f_{\kappa,\beta}$ (see \cite[Table 1]{Graf16}).

 {\textbf{End of proof.} Having proven that $\vol(I^+(\Sigma)) < \infty$, it directly follows that $\vol(I^+(x)) < \infty$ for all $x \in \Sigma \cup I^+(\Sigma)$ (in fact even for all $x \in M$). Since globally hyperbolic spacetimes are chronological, it follows by Theorem \ref{thm1} that $(M,g)$ is future volume incomplete.}
\end{proof}

Note that the proof works for arbitrary $\kappa$, $\beta$, as long as there is a uniform bound on the length of all geodesics starting from $\Sigma$. Therefore, in the class of spacetimes satisfying $\area(\Sigma) < \infty$ and the $\ccc$, a uniform bound on the length of all geodesics is strictly stronger than volume incompleteness. In particular, in the equality case $\kappa = -\left(\beta/n\right)^2$ {with $\beta < 0$}, there is a volume singularity but not all geodesics are incomplete: this is precisely what happens in the corresponding model space of constant curvature. In fact, it is a result of Andersson and Galloway \cite{AnGa} that this happens \emph{only} in the model space.

\begin{thm}[{Rigidity, \cite[Prop. 3.4]{AnGa}}]
 Let $(M,g)$ be a globally hyperbolic $(n+1)$-dimensional spacetime and $\Sigma \subset M$ a compact Cauchy hypersurface. Assume that $M$ and $\Sigma$ satisfy the $\ccc$ with $\beta \leq 0$ and $\kappa = -\left(\beta/n\right)^2$. If all geodesics normal to $\Sigma$ are future-complete, then $(J^+(\Sigma),g)$ is isometric to $[0,\infty) \times \Sigma$ with metric
 \[-dt^2 + \exp\left(-2 \sqrt{|\kappa|} t\right) h,\]
 where $h$ is a Riemannian metric on $\Sigma$, constant in $t$.
\end{thm}

Note that \cite[Prop. 3.4]{AnGa} is formulated for $\beta = -n$ (and time-reversed), but the proof can easily be adapted to arbitrary $\beta < 0$ (cf.\ \cite{GraSing}). The theorem was also generalized to Bakry--\'Emery spacetimes by Galloway and Woolgar \cite{GaWo}.

We believe that a volume singularity theorem with integral curvature bounds is also very much within reach by adapting geodesic versions due to Paeng \cite{PaeSing} and Graf et al \cite{GKOS}, both of which employ volume comparison techniques in their proof. Integral curvature bounds are weaker and more realistic than pointwise ones, especially when taking into account quantum properties of matter. In the future, we hope to prove volume singularity theorems with even weaker assumptions, in regimes where geodesic incompleteness might be too much to ask for.

\section{The cosmological volume function} \label{sec:voltime}

Andersson, Galloway, and Howard \cite{AGH98} introduced the notion of \emph{regular cosmological time function} for Big Bang spacetimes where all geodesics are past incomplete. It is defined as the supremum of the lengths of past directed causal curves starting at each point, and the adjective regular refers to the same properties 1 and 2 as in Definition \ref{def:volfct} below. Earlier and independently, Wald and Yip \cite{WaYi} had defined a similar time function, but without the regularity property. In this section, we propose a similar construction to that of Andersson, Galloway, and Howard for spacetimes with a past volume singularity. Our notion also bears resemblance to the volume functions of Geroch \cite{GerDD}, Hawking and Sachs \cite{HaSa} and Dieckmann \cite{Die}, but while they use an auxiliary finite measure, we use the canonical volume measure $\vol$ induced by the spacetime metric. Also Major, Rideout and Surya \cite{MRS} have a similar construction of a time function measuring the volume ``from a Cauchy surface", while we measure the volume ``from the Big Bang".

\begin{defn} \label{def:volfct}
The \emph{cosmological volume function} is defined as
\begin{equation*}
    \tau(x) := \vol (I^-(x)).
\end{equation*}
We say that $\tau$ is \emph{regular} if additionally
\begin{enumerate}
    \item $\tau(x) < \infty$ for all $x \in M$,
    \item $\tau \to 0$ along all past-inextendible causal curves.
\end{enumerate}
\end{defn}

Notice that there are spacetimes, such as Example \ref{ex:volin}, which have a regular cosmological volume function but not a regular cosmological time function in the sense of \cite{AGH98}. Both classes of functions share some key properties:

\begin{lem} \label{lem:tauiso}
Every regular cosmological volume function $\tau$ is {continuous and monotonously increasing along every future-directed causal curve.}
\end{lem}

{At the end of this section, we prove that $\tau$ is even strictly increasing, hence a time function.}

\begin{proof}
By transitivity of the chronological relation, we have that $\tau$ is increasing along every future-directed causal curve. Lower semicontinuity of $\tau$ follows from \cite[Prop.\ 1.4]{Die}, up to the caveat that the measures in \cite{Die} are finite, while $\vol$ is not (in general). However, by regularity of $\tau$, $\vol (I^-(x))$ is finite for every $x \in M$, which is all that is really needed.

By \cite[Prop.\ 1.6]{Die}, $\tau$ is upper semicontinuous if and only if $(M,g)$ is past reflecting. Thus our strategy is to assume that $(M,g)$ is not past reflecting, and to use this to construct a past-inextendible timelike curve $\tilde\gamma$ along which $\tau$ remains bounded below, contradicting Definition \ref{def:volfct}. The construction of $\tilde\gamma$ is similar to that in \cite[Prop.\ 2.1]{AGH98}.

For $(M,g)$ to not be past reflecting means that there exist points $x,y \in M$ such that $I^+(y) \subset I^+(x)$ but $I^-(x) \not\subset I^-(y)$. Let $(y_n)_n$ be a sequence of points in $I^+(y)$ that converges to $y$. Fix an auxiliary complete Riemannian metric $h$ on $M$. Since $I^+(y) \subset I^+(x)$, we may choose a sequence of past-directed causal curves $(\gamma_n)_n$ from $y_n$ to $x$ parametrized by $h$-arclength. Applying the limit curve theorem around $y$, we obtain a past-directed limit curve $\gamma$ starting at $\gamma(0) = y$ (up to passing to a subsequence, again denoted by $\gamma_n$). The curve $\gamma$ cannot reach $x$, since this would contradict our assumption that $I^-(x) \not\subset I^-(y)$. Hence $\gamma$ must be past-inextendible, and in particular parametrized on $[0,\infty)$.

The last step is to construct an inextendible past-directed timelike curve $\tilde \gamma \subset I^+(\gamma)$. It then follows that $\tau(\tilde\gamma) \geq \tau(x) >0$ by the following argument: Let $s \in [0, \infty)$. Then there exists a $t_s \in [0,\infty)$ such that $\gamma(t_s) \in I^-(\tilde\gamma(s))$. But then, by definition of $\gamma$ and openness of $I^-(\tilde\gamma(s))$, there exists an $n \in \nat$ such that $\gamma_n(t_s) \in I^-(\tilde\gamma(s))$. Since $\gamma_n(t_s) \in I^+(x)$, it follows that $\tau(\tilde\gamma(s)) \geq \tau(\gamma_n(t_s)) \geq \tau(x)$.

It remains to construct $\tilde\gamma$. Let $y_l := \gamma(l)$, and choose a $z_1 \in I^+ (y_1)$. By openness of $I^-(z_1)$, and since $y_2 \in I^-(y_1) \subset I^-(z_1)$, we may choose $z_2 \in I^-(z_1) \cap I^+(y_2) \cap B^h_{1/2}(y_2)$. Iterating this procedure, we obtain a sequence $(z_l)_l$ such that $z_l \in I^-(z_{l-1}) \cap I^+(y_l) \cap B^h_{1/l}(y_l)$. Then we construct $\tilde\gamma$ by joining all the past-directed timelike segments going from $z_l$ to $z_{l+1}$. Moreover, by contruction $\lim_{l \to \infty} z_l = \lim_{l \to \infty} y_l$, which does not exist, and hence $\tilde\gamma$ is indeed past-inextendible.

Summarizing, we have shown that if $(M,g)$ is not past-reflecting, then the cosmological volume function $\tau$ is not regular. Thus if $\tau$ is regular, $(M,g)$ must be past-reflecting, and then by \cite[Prop.\ 1.6]{Die}, $\tau$ is continuous.
\end{proof}

{We have seen that regularity of $\tau$ implies that $(M,g)$ must be past-reflecting. Next we show that it even forces $(M,g)$ to be globally hyperbolic. This is an interesting fact on its own, and will later aid us in proving that $\tau$ is a time function.}

\begin{thm} \label{thm:voltimegh}
If $(M,g)$ admits a regular cosmological volume function $\tau$, then $(M,g)$ is globally hyperbolic.
\end{thm}

\begin{proof}
{First we show that $(M,g)$ is causal. Assume it is not, meaning there is a closed causal curve $\gamma \colon I \to M$. Then, by the push-up lemma, $I^-(\gamma(s)) = I^-(\gamma(t))$ for all $s,t \in I$. It follows that $s \mapsto \tau(\gamma(s))$ is the constant function, contradicting regularity of $\tau$. Hence $(M,g)$ is causal, and it suffices to prove that the space of causal curves between any two points is compact in the $C^0$-topology in order to show global hyperbolicity \cite[Thm.\ 3.79]{MiSa}.}

{Let $y,z \in M$ be two points such that $y \in J^-(z)$, and let $(\gamma_n)_n$ denote any sequence of past-directed causal curves starting at $z$ and ending at $y$. By the limit curve theorem \cite[Thm.\ 2.53]{MinRev}, there exists a limit curve $\gamma_\infty$ starting at $z$, which either ends at $y$ or is past-inextendible. In the latter case, since by Lemma \ref{lem:tauiso} $\tau$ is continuous and $\tau \circ \gamma_n$ is bounded below by $\tau(y)$ for every $n$, also $\tau \circ \gamma_\infty$ is bounded below by $\tau(y) > 0$. This contradicts regularity of $\tau$, so we conclude that the limit curve $\gamma_\infty$ ends at $y$. We have thus shown that the space of causal curves between $y$ and $z$ is compact in the $C^0$-topology, and conclude that $(M,g)$ is globally hyperbolic.}
\end{proof}

\begin{thm}
 {Every regular cosmological volume function $\tau$ is a time function. Moreover, the level sets of $\tau$ are future Cauchy surfaces. If $\vol(I^-(\gamma(s))) \to \infty$ along all inextendible future directed causal curves $\gamma$, then the level sets are Cauchy surfaces.}
\end{thm}

\begin{proof}
 {Suppose that $x$ and $y$ are connected by a future-directed causal curve $\gamma : [0,1] \to M$. Then $I^-(x) \subseteq I^-(y)$ and hence $\tau(x) \leq \tau(y)$. By Lemma \ref{lem:subsetneq}, $\tau(x) = \tau(y)$ is only possible if $I^-(x) = I^-(y)$. Since by Theorem \ref{thm:voltimegh} the spacetime $(M,g)$ is globally hyperbolic, by the causal ladder \cite[Fig.\ 3.3]{BEE} it must be distinguishing, and therefore $I^-(x) = I^-(y)$ implies $x=y$. This concludes the proof that $\tau$ is a time function, since it follows that $\tau$ must be strictly increasing along (non-constant) future-directed causal curves, and continuity of $\tau$ was already established in Lemma \ref{lem:tauiso}.}

 {Next we show that the level sets $S_t := \{x \in M \mid \tau(x) = t\}$ are future Cauchy surfaces (cf.\ \cite[Prop.\ 2.2]{AGH98})}. Since $\tau$ is a time function, every $S_t$ is closed, acausal and edgeless. Moreover, since $\tau \to 0$ along every past-inextendible causal curve, it follows that $D^+(S_t) = \{x \in M \mid \tau(x) \geq t\} = J^+(S_t)$, which by defninition means that $S_t$ is future Cauchy. Here $D^+(S_t)$ is the future domain of dependence, meaning the set of all points $x \in M$ such that every past-inextendible causal curve starting at $x$ intersects $S_t$.

 Finally, if $\vol(I^-(\gamma(s))) \to \infty$ along all inextendible future directed causal curves $\gamma$, then for every such $\gamma$, $\tau(\gamma(s)) \to \infty$ towards the future. Since by regularity of $\tau$ also $\tau(\gamma(s)) \to 0$ towards the past, it follows by continuity of $\tau$ that every inextendible $\gamma$ intersects every level set $S_t$, and hence every $S_t$ is a Cauchy surface (for $0<t<\infty$).
\end{proof}

In fact, it suffices for the last statement to hold that $\lim \vol(I^-(\gamma(s)))$ attains the same (possibly finite) value for all $\gamma$, but unless the spacetime is very symmetric, this is only to be expected if the limit value is actually infinity. In any case, the fact that the level sets are future Cauchy, combined with \cite[Thm.~1.9~\&~Rem.~3.6]{BuGHnull}, yields the following corollary.

\begin{cor}
    Let $(M,g)$ be a spacetime with regular cosmological volume function $\tau$. Then the causal relation on $(M,g)$ is (globally) encoded in the null distance $\hat{d}_\tau$.
\end{cor}

The null distance was defined by Sormani and Vega \cite{SoVe} as a way to canonically metrize the topology on a spacetime with a time function (see also \cite{AlBu,SaSo}). The dependence on the choice of time function is undesired, and therefore the cosmological time function of Andersson, Galloway, and Howard is often cited as a canonical choice. We have thus shown that the cosmological volume function is a viable alternative. Another way to see this is the following: The null distance $\hat{d}_\tau$ is a conformal invariant. Thus, in order to uniquely recover the spacetime $(M,g)$ from $(M,\hat{d}_\tau,\tau)$, it is necessary that $\tau$ contains the information about the conformal factor. Since the conformal factor can be equivalently given by specifying the Lorentzian length functional $L_g$ \emph{or} the canonical volume element $\vol$, both the cosmological time function or the cosmological volume function are well suited for the task.

Finally, note that Andersson, Galloway, and Howard also show that their regular cosmological time functions are locally Lipschitz continuous, with well-defined first and second derivative at almost every point \cite[Thm.~1.2(v)]{AGH98}. It remains open if this is also true for regular cosmological volume functions. One case where it is likely to hold is when the spacetime obeys lower curvature bounds (i.e.\ energy conditions), since those imply bounds on the growth of volumes \cite{TrGr13} (see also proof of Theorem \ref{thm:sing}). {The differentiability of volume time functions in the traditional sense (i.e. for auxilary finite measures) has been investigated by Chru\'sciel, Grant, and Minguzzi, who also provided an example where what we here call the cosmological volume function is not $C^1$, but still Lipschitz \cite[Fig. 1]{CGM}.}

\section{Conclusions}

We have introduced the new notion of volume incompleteness as an alternative to geodesic incompleteness. The physical motivation is that if the future of a spacetime point has volume less than a Planck volume, then we expect quantum effects to become important for any prediction that an observer at that point would make, signaling a breakdown of General Relativity.

We have seen that the Schwarzschild and (non-maximally extended, sub-extremal) Kerr spacetimes are volume incomplete. More generally, volume singularities have the desirable feature that they are automatically hidden by a horizon. This leads us to conjecture that all singularities in realistic spacetimes should be volume singularities, as a new alternative way to formalize Penrose's cosmic censorship conjecture. In terms of future research, a very promising direction is laid down by Conjecture \ref{conj:Penvolsing}, that is, to prove a volume version of Penrose's singularity theorem which then automatically predicts the existence of an event horizon.

In the cosmological setting, we have argued that volume incompleteness should be seen as weaker than geodesic incompleteness, and have proven a basic Hawking-style singularity theorem. This opens the door to new singularity theorems in scenarios where proving geodesic incompleteness might not be possible. We have also defined a new type of time function for volume incomplete cosmological spacetimes, which is well-behaved in some situations where previous notions were not. Applications of the latter include the study of the null distance.

\bigskip
\noindent \textbf{Acknowledgements.} I am grateful to Greg Galloway, Melanie Graf, Eric Ling, Miguel S\'anchez, Elefterios Soultanis and Eric Woolgar for conversations at various stages of this work. Various parts of the paper were improved thanks to comments from anonymous reviewers. Part of this project was realized at, and supported by, the Fields Institute for Research in Mathematical Sciences, during the thematic program on Nonsmooth Riemannian and Lorentzian Geometry.

\newpage
\noindent \textbf{Declarations.} This version of the article has been accepted for publication, after peer review, but is not the Version of Record and does not reflect post-acceptance improvements, or any corrections. The Version of Record is available online at \\ \textcolor{blue}{\href{https://doi.org/10.1007/s11005-024-01814-y}{https://doi.org/10.1007/s11005-024-01814-y}}.

\printbibliography
\end{document}